\documentclass[conference,10pt,letterpaper]{IEEEtran}

\usepackage{amsmath,amssymb,amsthm}
\usepackage{xifthen}
\usepackage{mathtools}
\usepackage{enumerate}
\usepackage{microtype}
\usepackage{xspace}
\usepackage{bm}
\usepackage[T1]{fontenc}
\usepackage{lastpage}
\usepackage{bbm}
\usepackage[skipabove=10pt]{mdframed}

\newcommand{\mbs}[1]{\pmb{#1}}
\newcommand{\vect}[1]{{\lowercase{\mbs{#1}}}}
\newcommand{\mat}[1]{{\uppercase{\mbs{#1}}}}

\newcommand{\rv}[1]{{\mathsf{#1}}}
\newcommand{\rvVec}[1]{{\pmb{\mathsf{#1}}}}
\newcommand{\rvMat}[1]{{\pmb{\mathsf{#1}}}}

\newcommand{\T}{{\scriptscriptstyle\mathsf{T}}}
\renewcommand{\H}{{\scriptscriptstyle\mathsf{H}}}

\newcommand{\cond}{\,\vert\,}
\renewcommand{\Re}[1][]{\ifthenelse{\isempty{#1}}{\operatorname{Re}}{\operatorname{Re}\left(#1\right)}}
\renewcommand{\Im}[1][]{\ifthenelse{\isempty{#1}}{\operatorname{Im}}{\operatorname{Im}\left(#1\right)}}

\newcommand{\SNR}{\mathsf{snr}}

\newcommand{\vv}{\vect{v}}

\newcommand{\Sigmam}{\pmb{\Sigma}}

\newcommand{\Am}{\mat{a}}
\newcommand{\Bm}{\mat{b}}

\newcommand{\Mm}{\mat{M}}

\newcommand{\Qm}{\mat{q}}

\newcommand{\Um}{\mat{u}}
\newcommand{\Vm}{\mat{V}}

\newcommand{\Xm}{\mat{x}}

\newcommand{\Ic}{{\mathcal I}}

\newcommand{\Sc}{{\mathcal S}}

\newcommand{\CC}{\mathbb{C}}

\newcommand{\Id}{\mat{\mathit{I}}}

\newcommand{\CN}[1][]{\ifthenelse{\isempty{#1}}{\mathcal{N}_{\mathbb{C}}}{\mathcal{N}_{\mathbb{C}}\left(#1\right)}}

\renewcommand{\P}[1][]{\ifthenelse{\isempty{#1}}{\mathbb{P}}{\mathbb{P}\left(#1\right)}}
\newcommand{\E}[1][]{\ifthenelse{\isempty{#1}}{\mathbb{E}}{\mathbb{E}\left[#1\right]}}
\newcommand{\I}[1][]{\ifthenelse{\isempty{#1}}{\mathbb{I}}{\mathbb{I}\left\{#1\right\}}}
\renewcommand{\det}[1][]{\ifthenelse{\isempty{#1}}{\mathrm{det}}{\mathrm{det}\left(#1\right)}}
\newcommand{\trace}[1][]{\ifthenelse{\isempty{#1}}{\mathrm{tr}}{\mathrm{tr}\left(#1\right)}}
\newcommand{\rank}[1][]{\ifthenelse{\isempty{#1}}{\mathrm{rank}}{\text{rank}\left(#1\right)}}
\newcommand{\diag}[1][]{\ifthenelse{\isempty{#1}}{\mathrm{diag}}{\text{diag}\left(#1\right)}}
\newcommand{\Cov}[1][]{\ifthenelse{\isempty{#1}}{\mathsf{Cov}}{\mathsf{Cov}\left(#1\right)}}

\newcommand{\defeq}{\triangleq}
\newcommand{\eqdef}{\triangleq}

\newtheorem{remark}{Remark}
\newtheorem{theorem}{Theorem}
\newtheorem{corollary}{Corollary}
\newtheorem{lemma}{Lemma}

\IfFileExists{MinionPro.sty}{
}{
}

\newcommand{\Span}[1][]{\ifthenelse{\isempty{#1}}{\operatorname{Span}}{\operatorname{Span}\left(#1\right)}}

\newcommand{\floor}[1]{\left\lfloor #1 \right\rfloor}

\newcommand*\dif{\mathop{}\!\mathrm{d}}

\renewcommand{\SNR}{{\rm SNR}}

\renewcommand{\defeq}{:=}
\renewcommand{\eqdef}{=:}

\newcommand{\ind}[1]{{\mathbbm{1}\!\left\{#1\right\}}}

\usepackage{subfigure}
\usepackage{multirow}
\usepackage{array}
\usepackage{amsmath}
\usepackage[table,dvipsnames]{xcolor}
\usepackage{enumitem}
\usepackage{stfloats}
\usepackage[toc,acronym,nonumberlist,nopostdot]{glossaries}
\makeindex
\makenoidxglossaries
\usepackage{array,makecell,booktabs}
\usepackage{tabstackengine}
\newcolumntype{M}[1]{>{\centering\arraybackslash}m{#1}}
\usepackage{multicol,multirow}
\usepackage{tikz,pgfplots}
\usetikzlibrary{shapes}
\usetikzlibrary{plotmarks}
\usetikzlibrary{spy}
\newcommand{\minNT}{{\underline{S}}}
\newcommand{\maxNT}{{\overline{S}}}
\newcommand{\minMNT}{{L}}

\allowdisplaybreaks

\title{The Optimal DoF for the Noncoherent MIMO Channel with Generic Block Fading} 

\author{\IEEEauthorblockN{Khac-Hoang Ngo\IEEEauthorrefmark{1}, Sheng Yang\IEEEauthorrefmark{2}, Maxime Guillaud\IEEEauthorrefmark{3}} \vspace{.1cm}
	\IEEEauthorblockA{ \IEEEauthorrefmark{1}Department of Electrical Engineering, Chalmers University of Technology, 41296 Gothenburg, Sweden\\
		\IEEEauthorrefmark{2}Laboratory of Signals and Systems, CentraleSup\'elec, Paris-Saclay University, 91190 Gif-sur-Yvette, France \\
		\IEEEauthorrefmark{3}Mathematical and Algorithmic Sciences Laboratory, Huawei Technologies France, 92100 Boulogne-Billancourt, France}
	Emails: {ngok@chalmers.se, sheng.yang@centralesupelec.fr, maxime.guillaud@huawei.com} \vspace{-.4cm}
}

\mathtoolsset{showonlyrefs}

\begin{document}
	
\maketitle
\date{\today}
\begin{abstract}
  The high-SNR capacity of the noncoherent MIMO channel has been derived for the case of independent and identically distributed (IID)  Rayleigh block fading by exploiting the Gaussianity of the channel matrix. This implies the optimal degrees of freedom (DoF), i.e., the capacity pre-log factor. Nevertheless, as far as the optimal DoF is concerned, IID Rayleigh fading is apparently a sufficient but not necessary condition. In this paper, we show that the optimal DoF for the IID Rayleigh block fading channel is also the optimal DoF for a more general class of {\em generic} block fading channels, in which the random channel matrix has finite power and finite differential entropy. Our main contribution is a novel converse proof based on the duality approach.
%
\end{abstract}


\section{Introduction}
Multiple-input multiple-output~(MIMO) technology, consisting in transmitting and/or receiving with multiple antennas,
has been an efficient solution to exploit the extra spatial degrees of freedom~(DoF) in wireless communications. Under the ideal assumption that the
channel matrix is well conditioned and known to either end of the channel, it was shown that the capacity of a point-to-point MIMO
channel scales linearly with the number of antennas as $C = \min\left\{ M, N \right\} \log \SNR + O(1)$ in the high signal-to-noise ratio~(SNR)
regime, where $M$ and $N$ are the numbers of transmit and receive antennas, respectively \cite{Telatar1999capacityMIMO, Foschini}. The DoF, defined as the pre-log
of the capacity at high SNR, is $\min\left\{ M, N \right\}$ in this case. 
In practice, however, the channel matrix varies over time and is not known \emph{a priori}. Communication without \emph{a priori} channel state information~(CSI) is said to be \emph{noncoherent}.

In this paper, we consider a noncoherent $M\times N$ MIMO channel. Under stationary fast Rayleigh fading, i.e., the channel changes independently after each channel use, it was shown that the channel capacity scales double-logarithmically with the SNR in the single-input single-output ($M=N=1$) case~\cite{Taricco1997capacity_noCSI}. This result was then generalized to the MIMO case in~\cite{Moser}, where the authors showed that the capacity scales as
$C = \log\log \SNR + \chi(\rvMat{H}) + o(1)$ where $\rvMat{H}$ is the channel matrix and $\chi(\rvMat{H})$ is
called the fading number of the channel. This implies a zero DoF. Remarkably, the Rayleigh fading assumption was not needed in~\cite{Moser}. Instead, it was broadly assumed that the channel matrix has finite differential entropy and finite second moment. We refer to this fading model as {\em generic} fading. Under block fading, i.e., the channel matrix is assumed to remain constant during each coherence block of $T$ channel uses and varies independently between blocks, high-SNR approximations of the capacity have been derived for the Rayleigh fading case only~\cite{Marzetta1999capacity,Hochwald2000unitaryspacetime,ZhengTse2002Grassman,Yang2013CapacityLargeMIMO}. The optimal DoF was shown to be 
\begin{equation} \label{eq:opt_DoF_p2p_Rayleigh}
	d_{\rm opt} = M^*\Big(1-\frac{M^*}{T}\Big),
\end{equation}
with $M^* \!\defeq\! \min\left\{ M, N, \lfloor
T/2 \rfloor \right\}$,
and can be achieved either by well-designed space-time
modulations~\cite{Hochwald2000unitaryspacetime,ZhengTse2002Grassman,
Yang2013CapacityLargeMIMO}, or by simple training-based
strategies~\cite{Hassibi2003howmuchtraining}. The converse in these
works was based on the Rayleigh fading assumption, using
either a direct approximation at high SNR~\cite{Marzetta1999capacity,ZhengTse2002Grassman} or
a duality upper bound with a well chosen auxiliary output
distribution~\cite{Yang2013CapacityLargeMIMO}.

In this work, we generalize the {DoF} result of~\cite{Marzetta1999capacity,Hochwald2000unitaryspacetime,ZhengTse2002Grassman,Yang2013CapacityLargeMIMO} to the \textit{generic} fading model. 
Specifically, we prove that the {DoF} given in \eqref{eq:opt_DoF_p2p_Rayleigh} is also the optimal {DoF} under generic block fading. The main technical contribution of this paper lies in the converse proof. Leveraging the duality upper bound~\cite{Moser}, we carefully choose an auxiliary output distribution with which we derive a tight {DoF} upper bound.

The remainder of this paper is organized as follows. 
We present the channel model in Section~\ref{sec:model}, and then the main result and the achievablility in Section~\ref{sec:result}. The converse proof is given in Section~\ref{sec:p2p_converse}. Finally, we conclude the paper with a future perspective in Section~\ref{sec:conclusion}. The mathematical preliminaries for our analysis are provided in the appendix. 

{\it Notation:} For random quantities, we use
non-italic letters with sans-serif fonts, e.g., a scalar $\rv{x}$, a vector $\rvVec{v}$, and a matrix
$\rvMat{M}$. Deterministic quantities are denoted 
with italic letters, e.g., a scalar $x$, a vector $\pmb{v}$, and a
matrix $\pmb{M}$. 
The Euclidean norm is denoted by $\|\vv\|$ and the Frobenius norm by $\|\Mm\|_{\rm F}$. 
The trace, 
transpose and conjugate transpose of $\Mm$ are denoted $\trace\{\Mm\}$, 
$\Mm^\T$ and $\Mm^\H$, respectively.
$\{\lambda_i(\Mm)\}$ denote the eigenvalues 
 of $\Mm$ in decreasing order.
 We use $\diag[x_1,\dots,x_N]$ to denote the diagonal matrix with diagonal entries $x_1,\dots,x_N$, and 
$H(\cdot)$, $h(\cdot)$, and 
$D(\cdot\|\cdot)$ to denote the entropy, differential entropy, and 
Kullback-Leibler~(KL) divergence, respectively. Logarithms are in base $2$; $\ind{\cdot}$ is the indicator function;
$\log^+(x) \defeq \max\{\log(x),0\}$;
$(x)^+ \defeq \max\{x,0\}$; 
$\Gamma(x) = \int_{0}^{\infty}z^{x-1}e^{-z}dz$ is the Gamma function; $\Gamma_m(a) \defeq \pi^{m(m-1)/2} \prod_{k=1}^{m} \Gamma(a-k+1)$ is the complex multivariate Gamma function; $\Ic_{\mu,\nu} \defeq \int_{0}^\infty \frac{x^\mu}{(1+x^2)^\nu} \dif x$ (see Lemma~\ref{lem:integrable} in the appendix). 

\section{Channel Model} \label{sec:model}
We consider a MIMO channel consisting of a transmitter equipped with $M$ antennas and a receiver with $N$ antennas. The channel between the transmitter and the receiver is flat and block fading with coherence time of $T$ channel uses. That is, the channel matrix $\rvMat{H} \in \CC^{N \times  M}$ containing the fading coefficients from the $M$ transmit antennas to the $N$ receive antennas remains unchanged during each block of length $T$ and changes independently between blocks. 
The 
realizations of $\rvMat{H}$ are {\em unknown} to both the transmitter and the receiver. 
During a coherence block $b$, the received signal is
\begin{align}
\rvMat{Y}[b] = \rvMat{H}[b] \rvMat{X}[b] + \rvMat{Z}[b], \quad b = 1,2,\dots,
\end{align}
where $\rvMat{Z}[b] \in \CC^{N \times T}$ is the additive white Gaussian noise~(AWGN) with independent and identically distributed (IID)~$\CN(0,1)$ entries and $\rvMat{X}[b]$ is the transmitted signal satisfying  the power constraint
\begin{align} \label{eq:power_constraint_P2P}
\frac{1}{n_{\rm b}}\sum_{b = 1}^{n_{\rm b}}\|\rvMat{X}[b]\|_{\rm F}^2 \le P T,
\end{align}
where $n_{\rm b}$ is the number of blocks spanned by a codeword. 
The parameter $P$ is referred to as the SNR of the channel. Hereafter, we omit the block index $b$ whenever confusion is not likely.

Since the channel is block memoryless, the channel capacity is given by
$
C(P) = \frac{1}{T}\displaystyle\max_{p_{\rvMat{X}}: \; \E[\|\rvMat{X}\|_{\rm F}^2] \le P T} I(\rvMat{X}; \rvMat{Y})
$
bits per channel use. 
Then we say that 
$d_{\rm opt}$ is the optimal DoF with
$
d_{\rm opt} \defeq \lim\limits_{P\to\infty} \frac{C(P)}{\log P}.
$
We assume that the channel matrix $\rvMat{H}$ is drawn from a generic distribution satisfying the following conditions:
\begin{align} \label{eq:generic_fading}
h(\rvMat{H}) &> -\infty, \quad \E[\| \rvMat{H} \|_{\rm F}^2] < \infty. 
\end{align}%
That is, the channel matrix has finite differential entropy and finite second moment.
This class of fading model includes as a special case the IID Rayleigh fading model in which $\rvMat{H}$ contains IID $\CN[0,1]$ entries considered in~\cite{Marzetta1999capacity,Hochwald2000unitaryspacetime,ZhengTse2002Grassman,Yang2013CapacityLargeMIMO}.

For notational convenience, we define some parameters related to the channel's coherence time $T$, the number of transmit antennas $M$, and the number of receive antennas $N$ as $\minNT \defeq \min\{N,T\}$, $\maxNT \defeq \max\{N,T\}$, $\minMNT \defeq \min\{M,N,T\}$, and $M' \defeq \min\{M,N\}$ for future reference.


\section{Main Result: The Optimal DoF} \label{sec:result}
The optimal DoF of the noncoherent MIMO generic block fading channel described above is stated in Theorem~\ref{thm:dof_P2P:DoFregion}.
\begin{theorem} \label{thm:dof_P2P:DoFregion}
	For the noncoherent $M\times N$ MIMO channel in generic, flat, and block fading with coherence interval $T$, if $T = 1$, the optimal DoF is zero; otherwise, the optimal DoF is given by
	\vspace{-.2cm}
		\begin{align} \label{eq:opt_DoF_p2p}
		d_{\rm opt} = M^* \left( 1- \frac{M^*}{T} \right)
		\end{align}  
	with $M^* \defeq \min\{M,N,\lfloor T/2\rfloor\}$.
\end{theorem} 
The zero optimal DoF result for $T = 1$ (fast fading) has been shown in~\cite{Moser} and is included in Theorem~\ref{thm:dof_P2P:DoFregion} for completeness. In this case, the channel capacity scales double-logarithmically with the SNR.
\begin{corollary}
	In the single input and/or single output case ($\min\{M,N\} = 1$) or the $T = 2$ case, the optimal DoF is $d_{\rm opt} = 1-\frac{1}{T}$.
\end{corollary}
\begin{remark}
	Theorem~\ref{thm:dof_P2P:DoFregion} generalizes the optimal DoF of the noncoherent IID Rayleigh block fading channel given in~\cite{ZhengTse2002Grassman,Yang2013CapacityLargeMIMO}. This results show that the optimal DoF \eqref{eq:opt_DoF_p2p} holds even for non-Rayleigh fading channels as long as the channel matrix has finite differential entropy and finite power.
\end{remark}	


For $T > 1$, the optimal DoF is achieved by using only $M^*$ antennas and a simple pilot-based scheme: let the transmitter send pilot symbols in $M^*$ channel uses of a coherence block, and send data symbols in the remaining $T-M^*$ channel uses; the receiver estimates the channel based on the received pilot symbols and detects coherently the data symbols based on the channel estimate. A performance analysis of this pilot-based scheme following the same lines of \cite[Section V]{ZhengTse2002Grassman}, \cite{Hassibi2003howmuchtraining} shows that the DoF \eqref{eq:opt_DoF_p2p} is indeed achievable. 
We present next the converse proof.

\section{The Converse Proof} \label{sec:p2p_converse}
In this converse proof, we shall make use of the mathematical preliminaries (Lemmas~\ref{lemma:dof_P2P:transform_entropy}, \ref{lem:eig_AB}, and \ref{lem:integrable}) in the appendix. 
The channel input-output mutual information is expressed as
\begin{align}
I(\rvMat{X};\rvMat{Y}) &= h(\rvMat{Y}) - h(\rvMat{Y} \cond \rvMat{X}). \label{eq:mutual_infor}
\end{align}
By using Lemma~\ref{lemma:dof_P2P:transform_entropy} with
$\rvMat{W}=[\rvMat{H}\ \rvMat{Z}]$ and $\Am=\left[ 
\Xm \ \Id_T \right]^\T$ for each realization $\Xm$ of $\rvMat{X}$, the entropy $h(\rvMat{Y} | \rvMat{X})$ is given by 
\begin{align} \label{eq:bound_hY|X}
h(\rvMat{Y} | \rvMat{X}) &= N\E[\log\det(\Id_T + \rvMat{X}^\H \rvMat{X})] + \E\big[h(\breve{\rvMat{H}})\big],
\end{align}
where $\breve{\rvMat{H}}$ contains the first $T$ columns of $[\rvMat{H} \ \rvMat{Z}] \rvMat{U}_{[\rvMat{X} \ \Id_T]^\T}$ with $\rvMat{U}_{[\rvMat{X} \ \Id_T]^\T}$ being an $(M+T) \times (M+T)$ unitary matrix containing the left singular values of $[\rvMat{X} \ \Id_T]^\T$. In particular, under IID Rayleigh fading,  $\breve{\rvMat{H}}$ is an $N\times T$ Gaussian matrix with IID $\CN(0,1)$ entries, thus $h(\breve{\rvMat{H}}) = NT \log (\pi e)$. 

To bound $h(\rvMat{Y})$, we use the duality approach~\cite{Moser} as follows
\begin{align}
h(\rvMat{Y}) &= \E[-\log p_\rvMat{Y}(\rvMat{Y})] \\
&= \E[-\log q_\rvMat{Y}(\rvMat{Y})] - D(p_\rvMat{Y} \| q_\rvMat{Y}) \\
&\le \E[-\log q_\rvMat{Y}(\rvMat{Y})], \label{eq:dof_P2P_duality}
\end{align}
due to the nonnegativity of the KL divergence
$D(p_{\rvMat{Y}} \| q_{\rvMat{Y}})$. Here, the distribution $p_{\rvMat{Y}}$ is imposed by the input,
channel, and noise distributions, while $q_{\rvMat{Y}}$ is
any distribution in $\CC^{N \times T}$.
Note that a proper choice of $q_{\rvMat{Y}}$ 
is the key to a tight upper bound.
Let us consider
the singular value decomposition~(SVD) of $\rvMat{Y}$: 
\begin{align}
\rvMat{Y} = \rvMat{U}\rvMat{\Sigma} \rvMat{V}^\H,
\end{align}
where $\rvMat{U} \!\in\! \CC^{N\times \minNT}$ and $\rvMat{V} \!\in\! \CC^{T \times \minNT}$ are (truncated) unitary matrices, and $\rvMat{\Sigma} = \diag[\sigma_1, \dots, \sigma_\minNT]$ contains the singular values of $\rvMat{Y}$ in decreasing order. 
To make the SVD unique, we further assume that the diagonal elements of $\rvMat{V}$ are real and nonnegative~\cite{Yang2013CapacityLargeMIMO}. Then $\rvMat{U}$ belongs to the Stiefel manifold $\Sc(\CC^N,\minNT)$, while $\rvMat{V}$ belongs to a submanifold $\tilde{\Sc}(\CC^T,\minNT)$ of $\Sc(\CC^T,\minNT)$. 
The Jacobian of this SVD transformation is given by~\cite[App.~A]{ZhengTse2002Grassman}
\begin{align}
	\!\!\!J_{\maxNT,\minNT}(\sigma_1,\dots, \sigma_\minNT) &= \prod_{i=1}^{\minNT} \sigma_i^{2(\maxNT-\minNT)+1} \prod_{i<j}^{\minNT}(\sigma_i^2 - \sigma_j^2)^2 \label{eq:Jacobian} \\
	&=\prod_{i=1}^{\minNT} \sigma_i^{2(\maxNT-\minNT)+1} \prod_{i=1}^{\minNT} \prod_{j=i+1}^{\minNT} (\sigma_i^2 - \sigma_j^2)^2 \\
	&\le \prod_{i=1}^{\minNT} \sigma_i^{2(\maxNT-\minNT)+1} \prod_{i=1}^{\minNT} \sigma_i^{4(\minNT-i)} \\
	&= \prod_{i=1}^{\minNT} \sigma_i^{2T+2N - 4i + 1}, \label{eq:Jacobian_bound} \\
	&\eqdef \hat{J}_{\maxNT,\minNT}(\sigma_1,\dots, \sigma_\minNT)	
\end{align}
where the inequality is due to the decreasing order of $\sigma_1,\dots,\sigma_\minNT$.
We choose $q_{\rvMat{Y}}$ such that $\rvMat{U}$, $\rvMat{V}$, and $\rvMat{\Sigma}$ are mutually independent with the following distributions.
\begin{itemize}[leftmargin=*]
	\item Since the signal power is not captured in the singular vectors, as far as the DoF is concerned, the choice of distribution on the manifold for $\rvMat{U}$ and $\rvMat{V}$ can be arbitrary as long as $\E[-\log q_\rvMat{U}(\rvMat{U})]$ and $\E[-\log q_\rvMat{V}(\rvMat{V})]$ are finite. Here, for a closed-form expression, we let $\rvMat{U}$ and $\rvMat{V}$ be uniformly distributed in the Stiefel manifold $\Sc(\CC^N,\minNT)$ and submanifold $\tilde{\Sc}(\CC^T,\minNT)$, respectively. 
	That is, 
	\begin{align}
	q_\rvMat{U}(\Um) &= \frac{1}{|\Sc(\CC^N,\minNT)|} \ind{\Um \in \Sc(\CC^N,\minNT)}, \label{eq:qU}\\
	q_\rvMat{V}(\Vm) &= \frac{1}{|\tilde{\Sc}(\CC^T,\minNT)|} \ind{\Vm \in \tilde{\Sc}(\CC^T,\minNT)}, \label{eq:qV}
	\end{align}
	where the volumes of $\Sc(\CC^N,\minNT)$ and $\tilde{\Sc}(\CC^T,\minNT)$ are given by $|\Sc(\CC^n,m)| = \frac{2^m\pi^{mn}}{\Gamma_m(n)}$ and $|\tilde{\Sc}(\CC^n,m)| = \frac{|\Sc(\CC^n,m)|}{(2\pi)^m} = \frac{\pi^{m(n-1)}}{\Gamma_m(n)}$, respectively~\cite[Sec.~V]{Marques2008derivation}.
	
	
	\item On the other hand, the choice of $q_{\rvMat{\Sigma}}$ is crucial in deriving a tight DoF upper bound. Our choice is made so that, after taking the Jacobian of the SVD transformation into account, $\E[-\log q_\rvMat{Y}(\rvMat{Y})]$ depends on $\{\sigma_i^2\}$ only through $\E[\log(1+\sigma_i^2)]$, which can be straightforwardly upper bounded in terms of $\log P$. Specifically, we let the singular values of $\rvMat{Y}$ follow the distribution with the pdf
	\begin{align}
	q_{\sigma_1,\dots,\sigma_\minNT}(\sigma_1,\dots,\sigma_\minNT) 
	&= \beta  \frac{\hat{J}_{\maxNT,\minNT}(\sigma_1,\dots, \sigma_\minNT)}{\prod_{i=1}^{\minNT}(1+\sigma^2_i)^{\alpha_i}} \\
	&= \beta \prod_{i=1}^{\minNT} \frac{\sigma_i^{2T+2N-4i+1}}{(1+\sigma^2_i)^{\alpha_i}}, \label{eq:qSigma}
	\end{align}
	where $\beta$ is a scaling factor. Lemma~\ref{lem:integrable} implies that with $2\alpha_i = (2T+2N-4i+1) + 1 + \epsilon$, that is, $\alpha_i = T+N-2i+1 + \frac{\epsilon}{2}$, $i \in [\minNT]$, for any $\epsilon > 0$, then $\prod_{i=1}^{\minNT} \frac{\sigma_i^{2T+2N-4i+1}}{(1+\sigma^2_i)^{\alpha_i}}$ is integrable, i.e., there exists $\beta$ such that $\beta \prod_{i=1}^{\minNT} \frac{\sigma_i^{2T+2N-4i+1}}{(1+\sigma^2_i)^{\alpha_i}}$ is a pdf. Specifically, $\beta$ is given by 
	$\beta = \prod_{i=1}^{\minNT} \Ic^{-1}_{2T+2N-4i+1,\alpha_i}$.
	
\end{itemize}
Having specified $q_{\rvMat{Y}}$, we now proceed to compute $\E[-\log q_\rvMat{Y}(\rvMat{Y})]$ using the change of variables as
\begin{align}
&\E[-\log q_\rvMat{Y}(\rvMat{Y})] \notag \\
&= \E[-\log q_{\rvMat{U},\rvMat{\Sigma},\rvMat{V}}(\rvMat{U},\rvMat{\Sigma},\rvMat{V})] + \E[\log\left(J_{\maxNT,\minNT}(\sigma_1,\dots, \sigma_\minNT)\right)]\\
&= \E[-\log q_\rvMat{U}(\rvMat{U})] + \E[-\log q_\rvMat{V}(\rvMat{V})] \notag \\
&\quad + \E[-\log q_{\sigma_1,\dots,\sigma_\minNT}(\sigma_1,\dots,\sigma_\minNT)] \notag \\
&\quad+ \E[\log\left(J_{\maxNT,\minNT}(\sigma_1,\dots, \sigma_\minNT)\right)]. \label{eq:tmp815}
\end{align}
Plugging \eqref{eq:Jacobian_bound}, \eqref{eq:qU}, \eqref{eq:qV}, and \eqref{eq:qSigma} into \eqref{eq:tmp815}, we obtain
\begin{align} \label{eq:tmp486}
\E[-\log q_\rvMat{Y}(\rvMat{Y})] &\le \log |\Sc(\CC^N,\minNT)| + \log |\tilde{\Sc}(\CC^T,\minNT)| - \log \beta \notag \\
&\quad + \sum_{i=1}^{\minNT} \alpha_i \E[\log(1 + \sigma_i^2)].
\end{align}
Substituting the bounds of $h(\rvMat{Y})$ in \eqref{eq:bound_hY|X} and $h(\rvMat{Y} \cond \rvMat{X})$ in \eqref{eq:dof_P2P_duality}, \eqref{eq:tmp486} into \eqref{eq:mutual_infor}, we have the following bound
\begin{align}
&I(\rvMat{X};\rvMat{Y}) \notag \\
&\le \sum_{i=1}^{\minNT} \alpha_i \E[\log(1 + \sigma_i^2)] - N\E[\log\det(\Id_T + \rvMat{X}^\H \rvMat{X})] \notag \\
&\quad +  \log |\Sc(\CC^N,\minNT)| + \log |\tilde{\Sc}(\CC^T,\minNT)| - \log \beta - \E\big[h(\breve{\rvMat{H}})\big] \\
&= \underbrace{\sum_{i=1}^{\minNT} (\alpha_i - N) \E[\log(1 + \sigma_i^2)]}_{c_1} \notag \\
&\quad + \underbrace{N\Bigg(\sum_{i=1}^{\minNT}\E[\log(1 + \sigma_i^2)]  - \E[\log\det(\Id_T + \rvMat{X}^\H \rvMat{X})]\Bigg)}_{c_2} \notag \\
&\quad +  \log |\Sc(\CC^N,\minNT)| + \log |\tilde{\Sc}(\CC^T,\minNT)| - \log \beta \notag\\
&\quad - \E\big[h(\breve{\rvMat{H}})\big]. \label{eq:tmp422}
\end{align}
To proceed, we bound $c_1$ and $c_2$. For $c_2$, $\E[\log\det(\Id_T + \rvMat{X}^\H \rvMat{X})]$ is bounded in terms of the singular values of $\rvMat{Y}$ as follows
\begin{align}
&\sum_{i=1}^{\minNT} \E[\log(1 + \sigma_i^2)] \notag \\
&= \E[\log\det(\Id_N + \rvMat{Y} \rvMat{Y}^\H)] \\
&= \E[\log\det(\Id_N + \rvMat{H} \rvMat{X} \rvMat{X}^\H \rvMat{H}^\H \!+\! \rvMat{H} \rvMat{X} \rvMat{Z}^\H \!+\! \rvMat{Z}\rvMat{X}^\H \rvMat{H}^\H \!+\! \rvMat{Z} \rvMat{Z}^\H)] \\
&\le \E[{\log\det[{\Id_N \!+\! \E_{\rvMat{Z}}\big[\rvMat{H} \rvMat{X} \rvMat{X}^\H \rvMat{H}^\H \!+\! \rvMat{H} \rvMat{X} \rvMat{Z}^\H \!+\! \rvMat{Z}\rvMat{X}^\H \rvMat{H}^\H \!+\! \rvMat{Z} \rvMat{Z}^\H\big]}]}] \label{eq:tmp700}\\
&= \E[{\log\det((1+T)\Id_N + \rvMat{H} \rvMat{X} \rvMat{X}^\H \rvMat{H}^\H)}] \label{eq:tmp701}\\
&= \E[\log\det(\Id_M + (T+1)^{-1}\rvMat{X} \rvMat{X}^\H \rvMat{H}^\H\rvMat{H})] + N \log(T+1) \label{eq:tmp833}\\
&\le \sum_{i=1}^{\minMNT} \E[\log\Big(1 + \lambda_i\big((T+1)^{-1}\rvMat{X} \rvMat{X}^\H \rvMat{H}^\H\rvMat{H}\big)\Big)] +  N \log(T\!+\!1) \label{eq:tmp834} \\
&\le \sum_{i=1}^{\minMNT} \E[\log\big(1 \!+\! (T\!+\!1)^{-1} \lambda_i(\rvMat{X} \rvMat{X}^\H) \lambda_1(\rvMat{H}^\H\rvMat{H})\big)] \!+\! N \log(T\!+\!1) \label{eq:tmp835} \\
&\le \sum_{i=1}^{\minMNT} \E_{\rvMat{X}}\bigg[\log\bigg(1 + \lambda_i(\rvMat{X} \rvMat{X}^\H) \frac{\E_{\rvMat{H}}\big[\|\rvMat{H}\|_{\rm F}^2\big]}{T+1}\bigg)\bigg] + N \log(T\!+\!1) \label{eq:tmp836} \\
&\le \sum_{i=1}^{\minMNT} \E[\log\big(1 \!+\! \lambda_i(\rvMat{X} \rvMat{X}^\H)\big)] \!+\! \log^+\frac{\E\big[\|\rvMat{H}\|_{\rm F}^2\big]}{T+1} + N\log(T\!+\!1) \label{eq:tmp837} \\
&\le \sum_{i=1}^{\min\{M,T\}} \E[\log\big(1 + \lambda_i(\rvMat{X} \rvMat{X}^\H) \big)] + \log^+\frac{\E\big[\|\rvMat{H}\|_{\rm F}^2\big]}{T+1} \notag \\ &\quad+ N\log(T+1)   \\
&= \E[\log\det\big(\Id_T \!+\! \rvMat{X}^\H\rvMat{X} \big)] + \log^+\frac{\E\big[\|\rvMat{H}\|_{\rm F}^2\big]}{T+1} + N\log(T\!+\!1),\!\!
\end{align}
where \eqref{eq:tmp700} follows from Jensen's inequality since the $\log\det$ function is concave on the set of positive definite matrices; \eqref{eq:tmp701} holds because $\E[\rvMat{Z}] = \mathbf{0}$ and $\E[\rvMat{Z}\rvMat{Z}^\H] = T \Id_N$; \eqref{eq:tmp834} holds because the rank of $\rvMat{X} \rvMat{X}^\H \rvMat{H}^\H\rvMat{H}$ is upper bounded by $\minMNT \defeq \min\{M,N,T\}$; \eqref{eq:tmp835} follows from Lemma~\ref{lem:eig_AB}; \eqref{eq:tmp836} is due to $\lambda_1(\rvMat{H}^\H\rvMat{H}) \le \|\rvMat{H}\|_{\rm F}^2$ and Jensen's inequality; and \eqref{eq:tmp837} follows from 
\begin{align}
	\log(1 + ax) &\le \log(\max\{1,a\} + \max\{1,a\}x) \\
	&= \log(1+x) + \log^+a, \forall x \ge 0, a\ge 0.
\end{align}
Therefore, 
\begin{align}
c_2 \le N\log^+\frac{\E\big[\|\rvMat{H}\|_{\rm F}^2\big]}{T+1} + N^2\log(T+1). \label{eq:c2}
\end{align}
For $c_1$, we use Jensen's inequality to write
\begin{align}
c_1 &\le \sum_{i=1}^{\minNT} (\alpha_i - N)^+ \log\big(1 + \E[\sigma_i^2]\big) \\
&= \sum_{i=1}^{\minNT} (T - 2i+1 + \epsilon/2)^+ \log\big(1 + \E[\sigma_i^2]\big),
\label{eq:c1a}
\end{align}
where we recall that $\alpha_i = T+N-2i+1 + \frac{\epsilon}{2}$, $i \in [\minNT]$.
For $i = 1,\dots, \minMNT$, 
we bound $\E[\sigma_i^2]$ as 
\begin{align}
\E[\sigma_i^2] &\le \E[{\trace[\rvMat{Y}\rvMat{Y}^\H]}] \\
&= \E[{\trace[\rvMat{H} \rvMat{X} \rvMat{X}^\H \rvMat{H}^\H + \rvMat{H} \rvMat{X} \rvMat{Z}^\H + \rvMat{Z}\rvMat{X}^\H \rvMat{H}^\H + \rvMat{Z} \rvMat{Z}^\H]}] \label{eq:tmp721}\\
&=  \E[{\trace[\rvMat{X} \rvMat{X}^\H \rvMat{H}^\H \rvMat{H}]}] + \E[{\trace[\rvMat{Z} \rvMat{Z}^\H]}]\\
&= \sum_{i=1}^{\minMNT}\E[\lambda_i(\rvMat{X} \rvMat{X}^\H \rvMat{H}^\H \rvMat{H})] + NT \\
&\le \sum_{i=1}^{\minMNT}\E[\lambda_i(\rvMat{X} \rvMat{X}^\H) \lambda_1( \rvMat{H}^\H \rvMat{H})] + NT \label{eq:tmp462}\\
&\le \sum_{i=1}^{\min\{M,T\}}\E[\lambda_i(\rvMat{X} \rvMat{X}^\H)] \E[\|\rvMat{H}\|_{\rm F}^2] + NT \label{eq:tmp463}\\
&= \E[\|\rvMat{X}\|_{\rm F}^2] \E[\|\rvMat{H}\|_{\rm F}^2] + NT \\
&\le PT \E[\|\rvMat{H}\|_{\rm F}^2] + NT, \label{eq:tmp699}
\end{align}
where \eqref{eq:tmp462} follows from Lemma~\ref{lem:eig_AB}, and \eqref{eq:tmp463} is due to $L \le \min\{M,T\}$ and $\lambda_1( \rvMat{H}^\H \rvMat{H}) \le \|\rvMat{H}\|_{\rm F}^2$.
Thus 
\begin{align}
\!\!\!\!\!\!\!\!\log(1+ \E[\sigma_i^2]) &\le \log\left(PT \E[\|\rvMat{H}\|_{\rm F}^2] + NT + 1\right) \\
&= \log \left(P\Big(T \E[\|\rvMat{H}\|_{\rm F}^2] + \frac{NT+1}{P}\Big)\right) \\
&= \log P \!+\! \log T \!+\! \log \E[\|\rvMat{H}\|_{\rm F}^2] \!+\! o(1)  \label{eq:tmp717}
\end{align}
as $ P\to \infty$.
In the high-SNR regime, since the noise variance is bounded, the main contributor to the power of $\rvMat{Y}$ is $\rvMat{H} \rvMat{X}$, which has rank at most $\minMNT$. Thus, it is intuitive that the $\minNT - \minMNT$ smallest singular values of $\rvMat{Y}$ carry information about the noise only and are bounded. To see this, we follow the footsteps in \cite[p.~377]{ZhengTse2002Grassman} as follows. Since $\sigma_{\minMNT}, \dots, \sigma_{\minNT}$ are the $\minNT - \minMNT$ smallest singular values of $\rvMat{Y}$, for any $(N - \minMNT) \times N$ truncated unitary matrix $\Qm$, we have
$
\sum_{i=\minMNT + 1}^{\minNT} \sigma_i^2 \le \trace[\Qm \rvMat{Y} \rvMat{Y}^\H \Qm^\H].
$
We write $\rvMat{Z} = \rvMat{Z}_1 + \rvMat{Z}_2$, where $\rvMat{Z}_1$ is the projection of $\rvMat{Z}$ onto the subspace spanned by the row vectors of $\rvMat{H} \rvMat{X}$, and $\rvMat{Z}_2$ contains the perpendicular components. 
Since the subspace $\Span[\rvMat{H} \rvMat{X}]$ is independent of $\rvMat{Z}$, the total power in $\rvMat{Z}_2$ is $\E[\trace( \rvMat{Z}_2 \rvMat{Z}_2^\H)] = N(T-\minMNT)$. 
Since $\rvMat{H} \rvMat{X} + \rvMat{Z}_1$ has rank $\minMNT$, we can find a $(N - \minMNT) \times N$ truncated unitary matrix $\Qm_0$ such that $\Qm_0(\rvMat{H} \rvMat{X} + \rvMat{Z}_1) = \mathbf{0}$. Note that $\Qm_0$ is independent of $\rvMat{Z}_2$, thus
\begin{multline}
\E[\sum_{i=\minMNT + 1}^{\minNT} \sigma_i^2] \le \E[\trace(\Qm_0 \rvMat{Y} \rvMat{Y}^\H \Qm_0^\H)] \\= \E[\trace(\Qm_0 \rvMat{Z}_2 \rvMat{Z}_2^\H \Qm_0^\H)] = (N-\minMNT)(T-\minMNT).
\end{multline}
This implies that 
\begin{align}
\E[\sigma_i^2] \le (N-\minMNT)(T-\minMNT), \quad i  = \minMNT + 1,\dots,\minNT. \label{eq:tmp715}
\end{align}
Plugging \eqref{eq:tmp717} and \eqref{eq:tmp715} into \eqref{eq:c1a}, 
we get
\begin{align}
c_1 &\le \sum_{i=1}^{\minMNT} (T - 2i + 1 + \epsilon/2)^+ \big(\log P + \log \big(T\E[\|\rvMat{H}\|_{\rm F}^2]\big)\big) \notag \\
&\quad +  \sum_{i=\minMNT + 1}^{\minNT} (T - 2i + 1 + \epsilon/2)^+ \log\big(1+ (N-\minMNT)(T-\minMNT)\big)\!\! 
 \notag \\
 &\quad + o(1).
\label{eq:c1}
\end{align}
Substituting \eqref{eq:c2} and \eqref{eq:c1} into \eqref{eq:tmp422}, after some manipulations, we obtain that for any $\epsilon > 0$,
\begin{equation}
I(\rvMat{X};\rvMat{Y}) \le \sum_{i=1}^{\minMNT} (T - 2i + 1 + \epsilon/2)^+ \log P + c_0 + o(1)
\label{eq:rate_scaling}
\end{equation}
as $P \to \infty$, where 
\begin{align}
	c_0 &= \sum_{i=1}^{\minMNT} (T - 2i + 1 + \epsilon/2)^+ \log \big(T\E[\|\rvMat{H}\|_{\rm F}^2]\big) \notag \\
	&\quad +  \sum_{i=\minMNT + 1}^{\minNT} (T - 2i + 1 + \epsilon/2)^+ \log\big(1+ (N-\minMNT)(T-\minMNT)\big) \notag \\
	&\quad + \log |\Sc(\CC^N,\minNT)| + \log |\tilde{\Sc}(\CC^T,\minNT)| - \log \beta - \E\big[h(\breve{\rvMat{H}})\big] \notag \\
	&\quad + N\log^+\frac{\E\big[\|\rvMat{H}\|_{\rm F}^2\big]}{T+1} + N^2\log(T+1).
\end{align}
We see that the high-SNR capacity pre-log is $\sum_{i=1}^{\minMNT} (T - 2i + 1 + \epsilon/2)^+$. Letting $\epsilon$ arbitrarily close to zero (but remaining positive), this pre-log converges to 
$
	\sum_{i=1}^{\minMNT} (T - 2i + 1)^+ = \sum_{i=1}^{\min\{\minMNT,\lfloor T/2 \rfloor \}} (T \!-\! 2i \!+\! 1)
	\!=\! \sum_{i=1}^{M^*} (T \!-\! 2i \!+\! 1)
	\!=\! M^*(T\!-\!M^*),\!
$
where the first equality holds because $T - 2i + 1 < 0$ whenever $i> \floor{T/2}$.
Thus the optimal DoF is upper-bounded by $M^*\big(1-\frac{M^*}{T}\big)$. Furthermore, as $\epsilon \to 0$,
\begin{align}
c_0 &\to \minNT + \minNT(N+T-1)\log \pi - \log(\Gamma_\minNT(N) \Gamma_\minNT(T)) \notag \\
&\quad + \sum_{i=1}^{\minNT} \log \Ic_{2T+2N-4i+1,\alpha_i} + N^2 \log (T+1) - \E\big[h(\breve{\rvMat{H}})\big] \notag \\
&\quad + N\log^+\frac{\E\big[\|\rvMat{H}\|_{\rm F}^2\big]}{T+1} + M^*(T-M^*) \log\big(T\E[\|\rvMat{H}\|_{\rm F}^2]\big) \notag \\
&\quad + \ind{M'\le \lfloor T/2 \rfloor}\big(\lfloor T/2 \rfloor(T-\lfloor T/2 \rfloor) - M'(T-M')\big) \notag \\
&\qquad~ \times \log\big(1+ (N-M')(T-M')\big).
\end{align}


\section{Conclusion and Perspective} \label{sec:conclusion}
In this paper, we have derived the optimal DoF for the noncoherent MIMO generic block-fading channel. Our results generalize the known optimal DoF for the Rayleigh fading case to a wider class of fading in which the channel matrix has finite differential entropy and finite second moment. 


In the future, it would be interesting, as in the IID Rayleigh block fading case~\cite{ZhengTse2002Grassman,Yang2013CapacityLargeMIMO}, 
to characterize the constant term after the logarithmic term in the capacity formula.\footnote{In our analysis, the term $c_0$ in \eqref{eq:rate_scaling} would be a loose upper bound on the constant term in the channel capacity since the terms $\Ic_{2T+2N-4i+1,\alpha_i}$---although they do not scale with the power---become very large as $\epsilon \to 0$.}
Note that even for IID Rayleigh fading, no high-SNR approximation (up to a vanishing term) of the channel capacity has been found for the case $1 < T < 2\min\{M,N\}$. 
To this end, the \textit{escape-to-infinity} property~\cite{Moser,Durisi2011high} would be useful. It allows one to assume without loss of generality that the high-{SNR} capacity-achieving input distribution has no mass in a disk around the origin, whose radius can be made arbitrarily large. 

%

Our novel converse proof can be used for other problems, such as characterizing the optimal DoF region for the noncoherent MIMO multiple-access channel (MAC), which is not known even for the IID Rayleigh block fading case.
For the two-user single-input multiple-output (SIMO) MAC in generic block fading, we have found the optimal DoF region in~\cite{Hoang2018arXivDoFMAC}, but 
a generalization to the MIMO MAC was not obvious. 
The main challenge is to deal with inter-user interference which becomes an equivalent colored noise while decoding the signal of a user. This can be taken into account in the choice of auxiliary output distribution following the approach in the current paper.

%
\begin{appendices}
\section*{Appendix \\ Mathematical Preliminaries}
\begin{lemma}\label{lemma:dof_P2P:transform_entropy}
	Let $\Am = \Um \Sigmam \Vm \in\mathbb{C}^{m\times t}$ have full column rank ($\Um \in \CC^{m\times m}$), and
	$\rvMat{W}\in\mathbb{C}^{n\times m}$ be a random matrix such that
	$h(\rvMat{W})>-\infty$ and $\E[\|\rvMat{W}\|_{\rm F}^2] < \infty$. Then we have
	\begin{align}
		h(\rvMat{W}\Am) &= n \log\det(\Am^\H\Am) + h(\rvMat{W}'),
	\end{align}%
	where $\rvMat{W}'$ contains the first $t$ columns of $\rvMat{W}\Um$. Furthermore, $h(\rvMat{W}')$ is finite, i.e., $-\infty < h(\rvMat{W}') <\infty$.  
\end{lemma}
\begin{proof}
	See \cite[Appendix A-1]{Hoang2018arXivDoFMAC}.
\end{proof}

\begin{lemma} \label{lem:eig_AB}
	If $\Am$ and $\Bm$ are $n \times n$ Hermitian positive semidefinite matrices, then 
	\begin{align}
	\lambda_i(\Am) \lambda_n(\Bm) \le \lambda_i(\Am\Bm) \le \lambda_i(\Am) \lambda_1(\Bm), \quad i \in [n],
	\end{align}
where $\{\lambda_i(\Mm)\}$ denote the eigenvalues 
of a matrix $\Mm$ in decreasing order.
\end{lemma}
\begin{proof}
	The result follows immediately by applying \cite[Theorem 3]{Wang1992some} and \cite[Theorem 4]{Wang1992some} with $k=1$ therein.
\end{proof}

\begin{lemma} \label{lem:integrable}
	The function $f(x) = \frac{x^\mu}{(1+x^2)^\nu}, x \ge 0$ is integrable for any $\mu\ge 1$ and $2\nu > \mu+1$.
\end{lemma}
\begin{proof}
	Since $f(x)$ is a nonnegative function, we have $\int_0^\infty f(x) \dif x \le \int_0^\infty g(x) \dif x$ if $f(x)\le g(x), \forall x \ge 0$. 
	Let $\nu = \frac{\mu+1}{2} + \epsilon$ with $\epsilon > 0$, we have $\frac{x^\mu}{(1+x^2)^\nu} = \frac{x^\mu}{(1+x^2)^{\frac{\mu+1}{2}+\epsilon}}$ is a decreasing function in $\mu$ for any $x \ge 0$. Thus $\frac{x^\mu}{(1+x^2)^\nu} \le \frac{x}{(1+x^2)^{1+\epsilon}}$, $x \ge 0$, $\forall \mu\ge 1$. Therefore, $\int_{0}^\infty \frac{x^\mu}{(1+x^2)^\nu} \dif x \le \int_{0}^\infty \frac{x}{(1+x^2)^{1+\epsilon}} \dif x	 
	= \frac{1}{2\epsilon} < \infty$, $\forall \epsilon > 0.$
\end{proof}
We denote $\Ic_{\mu,\nu} \defeq \int_{0}^\infty \frac{x^\mu}{(1+x^2)^\nu} \dif x$. Note that $\Ic_{\mu,\nu} \to \infty$ as $2\nu - \mu - 1\to0$.

\end{appendices}
\bibliographystyle{IEEEtran}
\bibliography{IEEEabrv,./biblio}

\begin{thebibliography}{10}
\providecommand{\url}[1]{#1}
\csname url@samestyle\endcsname
\providecommand{\newblock}{\relax}
\providecommand{\bibinfo}[2]{#2}
\providecommand{\BIBentrySTDinterwordspacing}{\spaceskip=0pt\relax}
\providecommand{\BIBentryALTinterwordstretchfactor}{4}
\providecommand{\BIBentryALTinterwordspacing}{\spaceskip=\fontdimen2\font plus
\BIBentryALTinterwordstretchfactor\fontdimen3\font minus
  \fontdimen4\font\relax}
\providecommand{\BIBforeignlanguage}[2]{{%
\expandafter\ifx\csname l@#1\endcsname\relax
\typeout{** WARNING: IEEEtran.bst: No hyphenation pattern has been}%
\typeout{** loaded for the language `#1'. Using the pattern for}%
\typeout{** the default language instead.}%
\else
\language=\csname l@#1\endcsname
\fi
#2}}
\providecommand{\BIBdecl}{\relax}
\BIBdecl

\bibitem{Telatar1999capacityMIMO}
I.~Telatar, ``Capacity of multi-antenna {G}aussian channels,'' \emph{European
  Trans. Telecommun.}, vol.~10, pp. 585--595, Nov./Dec. 1999.

\bibitem{Foschini}
G.~J. Foschini and M.~J. Gans, ``On limits of wireless communications in a
  fading environment when using multiple antennas,'' \emph{Wireless personal
  communications}, vol.~6, no.~3, pp. 311--335, 1998.

\bibitem{Taricco1997capacity_noCSI}
G.~{Taricco} and M.~{Elia}, ``Capacity of fading channel with no side
  information,'' \emph{Electronics Letters}, vol.~33, no.~16, pp. 1368--1370,
  Jul. 1997.

\bibitem{Moser}
A.~Lapidoth and S.~Moser, ``Capacity bounds via duality with applications to
  multiple-antenna systems on flat-fading channels,'' \emph{IEEE Trans. Inf.
  Theory}, vol.~49, no.~10, pp. 2426--2467, Oct. 2003.

\bibitem{Marzetta1999capacity}
T.~L. Marzetta and B.~M. Hochwald, ``Capacity of a mobile multiple-antenna
  communication link in {R}ayleigh flat fading,'' \emph{IEEE Trans. Inf.
  Theory}, vol.~45, no.~1, pp. 139--157, Jan. 1999.

\bibitem{Hochwald2000unitaryspacetime}
B.~M. Hochwald and T.~L. Marzetta, ``Unitary space-time modulation for
  multiple-antenna communications in {R}ayleigh flat fading,'' \emph{IEEE
  Trans. Inf. Theory}, vol.~46, no.~2, pp. 543--564, Mar. 2000.

\bibitem{ZhengTse2002Grassman}
L.~Zheng and D.~N.~C. Tse, ``Communication on the {G}rassmann manifold: {A}
  geometric approach to the noncoherent multiple-antenna channel,'' \emph{IEEE
  Trans. Inf. Theory}, vol.~48, no.~2, pp. 359--383, Feb. 2002.

\bibitem{Yang2013CapacityLargeMIMO}
W.~Yang, G.~Durisi, and E.~Riegler, ``On the capacity of large-{MIMO}
  block-fading channels,'' \emph{IEEE J. Sel. Areas Commun.}, vol.~31, no.~2,
  pp. 117--132, Feb. 2013.

\bibitem{Hassibi2003howmuchtraining}
B.~Hassibi and B.~M. Hochwald, ``How much training is needed in
  multiple-antenna wireless links?'' \emph{IEEE Trans. Inf. Theory}, vol.~49,
  no.~4, pp. 951--963, Apr. 2003.

\bibitem{Marques2008derivation}
P.~M. {Marques} and S.~A. {Abrantes}, ``On the derivation of the exact,
  closed-form capacity formulas for receiver-sided correlated {MIMO}
  channels,'' \emph{IEEE Trans. Inf. Theory}, vol.~54, no.~3, pp. 1139--1161,
  2008.

\bibitem{Durisi2011high}
G.~Durisi and H.~B{\"o}lcskei, ``High-{SNR} capacity of wireless communication
  channels in the noncoherent setting: A primer,'' \emph{AEU-International
  Journal of Electronics and Communications}, vol.~65, no.~8, pp. 707--712,
  Aug. 2011.

\bibitem{Hoang2018arXivDoFMAC}
K.-H. Ngo, S.~Yang, and M.~Guillaud, ``The optimal {DoF} region for the
  two-user non-coherent {SIMO} multiple-access channel,'' \emph{arXiv preprint
  arXiv:1806.04102}, 2018.

\bibitem{Wang1992some}
B.~Wang and F.~Zhang, ``Some inequalities for the eigenvalues of the product of
  positive semidefinite {H}ermitian matrices,'' \emph{Linear algebra and its
  applications}, vol. 160, pp. 113--118, 1992.

\end{thebibliography}
\end{document}